%% file: ms.tex
\documentclass[letterpaper, 10 pt, journal]{IEEEtran}

\usepackage{cite}
\usepackage{graphicx}
\usepackage{amsfonts}
\usepackage{amsmath} 
\usepackage{amssymb}

\usepackage{amsthm}	
\usepackage{dsfont}
\usepackage{enumerate}
\usepackage{mathrsfs}
\usepackage{cleveref}
\usepackage[labelformat=simple]{subcaption}

\usepackage[hyphens]{url}

\usepackage{epstopdf}
\usepackage{tikz}
\usetikzlibrary{arrows}

\usepackage{nomencl}
\makenomenclature
\usepackage{etoolbox}
\renewcommand\nomgroup[1]{%
  \item[\bfseries
  \ifstrequal{#1}{A}{Acronyms}{%
  \ifstrequal{#1}{S}{Sets}{%
  \ifstrequal{#1}{P}{Parameters}{%
  \ifstrequal{#1}{V}{Variables}{}}}}%
]}

    \nomenclature[S, 02]{$\Ncal \subseteq \Nset$}{Set of EVs}
    \nomenclature[S, 01]{$\Tcal \subseteq \Nset$}{Set of time periods}
    \nomenclature[S, 03]{$\Uset_i \subseteq \Rset^T$}{Individual flexibility set of EV $i \in \Ncal$}
    \nomenclature[S, 04]{$\Uset = \sum_{i=1}^N \Uset_i$}{Aggregate flexibility set}
    \nomenclature[S, 05]{$\Uset_0 \subseteq \Rset^T$}{Base set}
    \nomenclature[S, 06]{$\Pset \subseteq \Rset^T$}{Inner approximation of $\Uset$}
    \nomenclature[A, 01]{EV}{Electric vehicle}
    \nomenclature[A, 02]{DER}{Distributed energy resource}
    \nomenclature[A, 03]{SoC}{State-of-charge}
    \nomenclature[A, 04]{TCL}{Thermostatically controlled load}
    \nomenclature[A, 04]{ISO}{Independent system operator}
    \nomenclature[A, 05]{H-representation}{Half-space representation}
    \nomenclature[A, 06]{H-polytope}{Polytope in H-representation}
    \nomenclature[A, 07]{AH-polytope}{Affine transformation of H-polytope}
    \nomenclature[A, 08]{V-representation}{Vertex representation}
    \nomenclature[A, 09]{V-polytope}{Polytope in V-representation}
    \nomenclature[P, 01]{$N \in \Nset$}{ Number of EVs}
    \nomenclature[P, 02]{$T \in \Nset$}{ Number of time periods}
    \nomenclature[P, 03]{$\delta \in \Rset_+$}{ Length of time period}
    \nomenclature[V, 01]{$u_i \in \Rset^T$}{Charging power profile of EV $i \in \Ncal$}
    \nomenclature[V, 02]{$\overline{u}_i \in \Rset^T$}{Upper limit on $u_i$}
    \nomenclature[V, 03]{$\underline{u}_i \in \Rset^T$}{Lower limit on $u_i$}
    \nomenclature[V, 01]{$u \in \Rset^T$}{Aggregate charging power profile}
    \nomenclature[V, 04]{$x_i \in \Rset^T$}{Net-energy profile of EV $i \in \Ncal$}
    \nomenclature[V, 05]{$\overline{x}_i \in \Rset^T$}{Upper limit on $x_i$}
    \nomenclature[V, 06]{$\underline{x}_i \in \Rset^T$}{Lower limit on $x_i$}
    \nomenclature[V, 08]{$\overline{p}, \, \gamma_i \in \Rset^T$}{Translation vectors}
    \nomenclature[V, 09]{$P, \,  \Gamma_i \in \Rset^{T\times T}$}{Linear transformation matrices}
    \nomenclature[V, 07]{$\alpha \in \Rset_+$}{Scaling factor}
    
\graphicspath{{figures/}}

\usepackage{geometry}
\input{header.tex}

\newcommand{\indi}[1]{\mathds{1}  \hspace{-.03in} \left\{#1 \right\}}

\newcommand{\vol}[1]{\text{Vol}\left(#1\right)}

\IEEEoverridecommandlockouts

\begin{document}

\title{An Efficient Method for Quantifying the Aggregate Flexibility of Plug-in Electric Vehicle Populations}

\author{Feras Al Taha, \, Tyrone Vincent, \,   Eilyan Bitar 
\thanks{This work was supported in part by the Natural Sciences and Engineering Research Council of Canada, in part by the Cornell Atkinson Center for Sustainability, in part by The Nature Conservancy, in part by the Holland Sustainability Project Trust, and in part by the Office of Naval Research via grant N00014-17-1-2697.}
\thanks{F. Al Taha is with the School of Electrical and Computer Engineering, Cornell University, Ithaca, NY, 14853, USA (e-mail:  foa6@cornell.edu). }%
\thanks{T. Vincent is with the Department of Electrical Engineering and Computer Science, Colorado School of Mines, Golden, CO 80401, USA (e-mail: tvincent@mines.edu).}%
\thanks{E. Bitar is with the School of Electrical and Computer Engineering, Cornell University, Ithaca, NY, 14853, USA (e-mail: eyb5@cornell.edu). }%
}

\maketitle 

\begin{abstract}

Plug-in electric vehicles (EVs) are widely recognized as being highly flexible electric loads that can be pooled and controlled via aggregators to provide low-cost energy and ancillary services to wholesale electricity markets.  To participate in these markets, an aggregator must encode the aggregate flexibility of the population of EVs under their command as a single polytope that is compliant with existing market rules. To this end, we investigate the problem of characterizing the aggregate flexibility set of a heterogeneous population of EVs whose individual flexibility sets are given as convex polytopes in half-space representation. As the exact computation of the aggregate flexibility set---the Minkowski sum of the individual flexibility sets---is known to be intractable, we study the problem of computing  maximum-volume inner approximations  to the aggregate flexibility set by optimizing over affine transformations of a given convex polytope in half-space representation.  We show how to conservatively approximate these set containment problems as linear programs that scale polynomially with the number and dimension of the individual flexibility sets.  The inner approximation methods provided in this paper generalize and improve upon existing methods from the literature. We illustrate the improvement in approximation accuracy and performance achievable by our methods with numerical experiments.

\end{abstract}

\input{Introduction}
\input{model} 
\input{Inner_Approximation}
\input{Experiments}
\input{Conclusion}

\input{Appendix}

\bibliographystyle{IEEEtran}
\bibliography{references}{\markboth{References}{References}}

\end{document}

%% file: header.tex
\newtheorem{thm}{Theorem}
\newtheorem{lem}[thm]{Lemma}

\newtheorem{defn}{Definition}

\newtheorem{exmp}{Example}
\newtheorem{rem}{Remark}

\newcommand{\R}{{\mathbb R}}

\newcommand{\trace}[1]{\text{Tr}\left(#1\right)}

\renewcommand{\arraystretch}{0.9}

\newcommand{\Nset}{\mathbb{N}}

\newcommand{\Pset}{\mathbb{P}}
\newcommand{\Rset}{\mathbb{R}}

\newcommand{\Uset}{\mathbb{U}}

\newcommand{\Xset}{\mathbb{X}}
\newcommand{\Yset}{\mathbb{Y}}
\newcommand{\Zset}{\mathbb{Z}}

\newcommand{\Ncal}{{\mathcal{N}}}

\newcommand{\Scal}{{\mathcal{S}}}
\newcommand{\Tcal}{{\mathcal{T}}}

\newcounter{l1}
\newcounter{l2}
\newcounter{l3}
\newcommand{\bdotlist}{\begin{list}{$\bullet$}{}}
\newcommand{\bboxlist}{\begin{list}{$\Box$}{}}
\newcommand{\bbboxlist}{\begin{list}{\raisebox{.005in}{{\tiny
$\blacksquare$ \ \ }}}{}}
\newcommand{\bdashlist}{\begin{list}{$-$}{} }
\newcommand{\blist}{\begin{list}{}{} }
\newcommand{\barablist}{\begin{list}{\arabic{l1}}{\usecounter{l1}}}
\newcommand{\balphlist}{\begin{list}{(\alph{l2})}{\usecounter{l2}}}
\newcommand{\bAlphlist}{\begin{list}{\Alph{l2}.}{\usecounter{l2}}}
\newcommand{\bdiamlist}{\begin{list}{$\diamond$}{}}
\newcommand{\bromalist}{\begin{list}{(\roman{l3})}{\usecounter{l3}}}

%% file: Introduction.tex
\section{Introduction} \label{sec:introduction}

The widescale electrification of the transportation sector will present both challenges and novel opportunities for the efficient and reliable operation of the power grid. 
In particular, the increase in electricity demand driven by plug-in electric vehicle (EV) charging will  sharply increase  peak demand if left unmanaged \cite{alexeenko2021achieving}.
However, a number of field studies have shown that the charging requirements of EVs are usually flexible in the sense that most EVs charging in workplace or residential settings remain connected to their chargers long after they have finished charging \cite{alexeenko2021achieving, lee2021adaptive, smart2015plugged, bauman2016residential}.   This flexibility can be utilized by coordinating the charging profiles of individual EVs to minimize their collective contribution to peak load, or to provide energy and/or ancillary services to the regional wholesale electricity market \cite{satchidanandan2022economic}.  

Indeed, enabled by regulations such as FERC Order No. 2222 \cite{ferc2222}, aggregators\footnote{An aggregator is an electricity market participant that may combine distributed energy resources across a wide range of types and sizes to participate in the market as a single entity, which is typically called an aggregation.} can pool and coordinate the control of multiple EVs and other distributed energy resources (DERs) to participate alongside conventional resources in the wholesale market \cite{henriquez2017participation}. 
Aggregators that wish to participate in the wholesale market must represent the individual flexibility sets of participating EVs as a single \emph{aggregate flexibility set} that accurately captures the supply/demand capabilities of the individual EVs as a collective. 
Crucially, these aggregate flexibility sets must be encoded using bid/offer formats that are compliant with existing market rules.  
Traditionally, bid/offer formats have been structured to reflect the supply and demand characteristics of conventional generators and load-serving entities. 
More recently, electricity market designs have evolved to incorporate aggregator bid/offer formats that more accurately capture the intertemporal supply and demand capabilities of energy storage resources, e.g., in the form of time-varying upper and lower limits on power, ramping, and battery state-of-charge (SoC) \cite{nyiso, caiso}.\footnote{Examples of such market designs include the New York ISO DER and aggregation participation model \cite{nyiso} and the  California ISO DER provider model \cite{caiso}.}
To participate in wholesale electricity markets, an aggregation of EVs must be offered into the marketplace and settled as a single resource that encapsulates the collective capacity of the EV aggregation to produce and consume energy over a fixed window of time while accounting for the individual charging needs of the participating EVs.
With this motivation in mind, we investigate the problem of designing efficient optimization-based methods to accurately approximate the aggregate flexibility of a finite population of EVs  as a single representative energy storage resource.

\subsection{Related Work}
The individual flexibility sets associated with a wide variety of distributed energy resources, including thermostatically controlled loads and plug-in EVs, are typically encoded as convex polytopes in  half-space representation \cite{zhao2017geometric, hao2014aggregate, barot2017concise}. The exact calculation of their aggregate flexibility set---the Minkowski sum of the individual flexibility sets---is  known to be computationally intractable in general \cite{gritzmann1993minkowski, trangbaek2012exact}. 
As a result, a variety of methods have been developed to efficiently compute approximations that are subsets or supersets of the aggregate flexibility set (termed \textit{inner} and \textit{outer} approximations, respectively) \cite{alizadeh2014capturing, barot2017concise,kundu2018approximating}. 
A shortcoming of outer approximations is that they may contain infeasible points. In contrast, inner approximations are guaranteed to only contain feasible points---a crucial property for control~applications.

There are a number of papers that provide closed-form inner approximations for the aggregate flexibility set as a function of the individual load parameters \cite{hao2014characterizing, hao2014aggregate, madjidian2017energy, nayyar2013aggregate}. While these inner approximations are trivial to compute, they have been observed to be very conservative when there is considerable heterogeneity between the individual flexibility sets \cite{zhao2017geometric}. There have been a number of attempts to utilize convex optimization methods to construct more accurate inner approximations.
For example, Zhao \emph{et al.} \cite{zhao2016extracting} approximate the aggregate flexibility set as a projection of a convex polytope. 

Another group of papers provide methods to approximate the aggregate flexibility set by constructing convex inner approximations of the individual flexibility sets using specific convex geometries that permit the efficient computation of their Minkowski sum. 
For example, M{\"u}ller \emph{et al.}  \cite{muller2017aggregation}  approximate the individual flexibility sets using a specific class of  zonotopes  (a family of centrally symmetric polytopes), while Zhao \emph{et al.} \cite{zhao2017geometric} utilize homothets (dilation and translation) of a user-defined convex polytope.  Nazir \emph{et al.} \cite{nazir2018inner} provide an algorithm to internally approximate the individual flexibility sets using unions of  homothets of axis-aligned hyperrectangles.  
While this algorithm can approximate the true Minkowski sum with arbitrary precision, ensuring high accuracy may require a large number of  hyperrectangles---potentially limiting scalability to large-scale systems. 
Furthermore, the resulting approximation to the aggregate flexibility set may not be compliant with the class of aggregation models mandated by ISO-administered markets, e.g., in the form of a singular energy storage resource.

We also note that there have been recent attempts to construct convex approximations of the aggregate flexibility set when the individual flexibility sets may be nonconvex \cite{hreinsson2021new, taheri2022data}. However, the approximations provided by these methods may contain infeasible points, as they are not provable inner approximations of the aggregate flexibility set.

\subsection{Main Contributions}
In this paper, we study the problem of computing  a \emph{maximum-volume inner approximation} of the aggregate flexibility set by optimizing over affine transformations of a given convex polytope in half-space representation. 
The proposed class of approximations generalizes those considered by related methods from the literature, which either limit the class of approximating sets to homothets of a given convex polytope \cite{zhao2017geometric} or restrict the specification of the given polytope to zonotopic geometries~\cite{muller2017aggregation}.  
Importantly, the optimization methods proposed in this paper can be used to construct an inner approximation of the aggregate flexibility set that is structured as a singular energy storage resource. 
This ensures compliance with electricity market rules that require an aggregation of multiple DERs to be offered into the marketplace as a single
resource that accurately captures the  operating range of the aggregation. 

The approach taken in this paper draws inspiration from the methods proposed in \cite{zhao2017geometric}.
Using standard techniques from convex analysis, we show how to conservatively approximate the maximum-volume inner approximation problem as a linear program that scales polynomially with the number and dimension of the individual flexibility sets. By considering a more general family of approximating polytopes (i.e., affine transformations of convex polytopes), we are able to efficiently compute approximations to the aggregate flexibility set that  improve upon the accuracy of those generated by the methods proposed in \cite{zhao2017geometric}.
We provide a stylized example (in Figure \ref{fig:2d-example}) and conduct numerical experiments (in Section \ref{sec:experiments}) that demonstrate the improvement in approximation accuracy of the proposed methods when compared to the methods proposed in \cite{zhao2017geometric} and \cite{muller2017aggregation}. 
We also show how to efficiently disaggregate any point within the proposed inner approximation of the aggregate flexibility set into a collection of individually feasible charging profiles using an affine mapping that is computed as a byproduct of the inner approximation method.

We note that, while we have focused on plug-in EVs as the motivating application for our analysis, the techniques developed in this paper can also be used to approximate the aggregate flexibility of other distributed energy resources whose individual flexibility sets can be expressed or approximated by convex polytopes in half-space representation. These include thermostatically controlled loads (TCLs) \cite{hao2014aggregate, biegel2013primary,mathieu2013state}, HVAC systems \cite{hao2014hvac}, and residential pool pumps~\cite{meyn2014ancillary}. 

\subsection{Notation} 
 We employ the following notational conventions throughout the paper.
 Let $\Rset$ and $\Zset$ denote the sets of real numbers and  integers, respectively. 
 We denote the indicator function of set $\Scal$ by $\indi{x\in\Scal} = 1$ if $x \in \Scal$ and $\indi{x\in\Scal} = 0$ if $x \notin \Scal$.  We denote the $n \times n$ identity matrix by $I_n$. Given a pair of matrices $A$ and $B$ of appropriate dimension, we let $(A, \, B)$ denote the matrix formed by stacking $A$ and $B$ vertically.  Given a vector $\gamma$ and matrix $\Gamma$ of appropriate dimension, we denote an affine transformation of a  set $\Xset$ by $\gamma+\Gamma\Xset := \{\gamma+\Gamma x \,\mid \, x\in\Xset\}$.
 
\subsection{Paper Organization}
The remainder of the paper is organized as follows. 
In Section \ref{sec:model}, we present the aggregate flexibility model and state the problem addressed in this paper. 
Linear programming-based methods to compute inner approximations to the aggregate flexibility set are derived in  Section \ref{sec:inner approximation results}.  
In Section \ref{sec:disaggregation}, we provide an efficient method to disaggregate charging profiles belonging to the proposed inner approximations of the aggregate flexibility set.
Numerical experiments illustrating the proposed approximation methods are provided in Section \ref{sec:experiments}. 
Section~\ref{sec:conclusion} concludes the paper. 
A list of commonly used abbreviations and symbols is provided in Appendix \ref{app:nomenclature}.

%% file: model.tex
\section{Problem Formulation}

\label{sec:model}
In this section, we present the model of individual EV flexibility sets and formulate the problem of finding maximum volume inner  approximations of the aggregate flexibility set. 
We consider a system in which an aggregator seeks to centrally manage the charging profiles of a finite population of plug-in electric vehicles (EVs) indexed by $i \in \Ncal:=  \{1, \dots, N\}$. Time is assumed to be discrete with periods indexed by $ t \in \Tcal := \{0, \dots, T-1\}$. All time periods are assumed to be of equal length, which we denote by $\delta > 0$.  

\subsection{EV Charging Dynamics}
We let $u_i(t)$ denote the charging rate of EV  $i \in \Ncal$ at time $t \in \Tcal$, and let the vector $u_i := (u_i(0), \, \dots, \, u_i(T-1)) \in \Rset^T$ denote its charging profile.
Given a charging profile $u_i$, the net energy supplied to each  EV  $i \in \Ncal$ is assumed to evolve according to the difference equation
\begin{align} \label{eq:dyn}
x_i(t+1) = x_i(t) + u_i(t) \delta , \quad t\in \Tcal,
\end{align}
where $x_i(0) = 0$ and $x_i(t)$ represents the net  energy delivered to EV $i$ over the previous $t$ time periods.
We denote the resulting net energy profile by the vector $x_i = (x_i(1), \dots, x_i(T)) \in \Rset^{T}$, which satisfies Eq. \eqref{eq:dyn} or, more concisely, the  relationship
\begin{align*}
 x_i = Lu_i,
\end{align*} 
where $ L \in \Rset^{T \times T}$ is a lower triangular matrix given by $L_{ij} := \delta$ for all $j \leq i$. 
The matrix $L$ is invertible since the elements along its diagonal are all non-zero. 
This lossless model of the EV charging dynamics is commonly used in the context of aggregate flexibility modeling, e.g., see \cite{muller2017aggregation,madjidian2017energy, hao2014characterizing, nayyar2013aggregate, zhao2016extracting}.

\subsection{Individual Flexibility Sets}
We refer to the nonempty set of admissible charging profiles associated with each EV $i \in \Ncal$ as an  \textit{individual flexibility set} and denote it by
\begin{align} \label{eq:indi_set}
\Uset_i : = \left\{ u \in \Rset^T \, | \, u \in [\underline{u}_i, \, \overline{u}_i], \ Lu \in [\underline{x}_i, \, \overline{x}_i]   \right\}.
\end{align}
The vectors $\underline{u}_i,  \overline{u}_i \in \Rset^T$ represent minimum and maximum power limits on the charging profile, respectively. The vectors  $\underline{x}_i, \overline{x}_i \in \Rset^T$ represent minimum and maximum energy limits on the net energy profile, respectively.  
Flexibility sets defined according to \eqref{eq:indi_set} are widely studied in the literature and referred to as \emph{virtual batteries} \cite{zhao2017geometric,hreinsson2021new} or \emph{generalized battery models} \cite{hao2014aggregate,muller2017aggregation,barot2017concise,nazir2018inner}.
The family of generalized battery models \eqref{eq:indi_set} is quite expressive. 
For example, it is able to capture time-varying EV power and energy constraints, net energy requirements, charging completion deadlines, and allowable/forbidden charging times that reflect an EV's connection status over time. 
Furthermore, the family of generalized battery models is compatible with the aggregator models currently used in wholesale electricity markets \cite{caiso,nyiso}, e.g., in the form of time-varying upper and lower limits on power and battery state-of-charge (SoC).

Note that each flexibility set $\Uset_i$ is a compact, convex polytope since it is defined as the intersection of $4T$ closed half-spaces that form a bounded set. 
It will be convenient to use a more concise expression for the individual flexibility sets given by
\begin{align} \label{eq:power_rep}
\Uset_i  = \left\{ u \in \Rset^T \, | \,   H u \leq h_i   \right\},
\end{align}
where $H := (L, -L, \, I_T, -I_T)$ and $h_i := (\overline{x}_i, -\underline{x}_i,  \overline{u}_i, -\underline{u}_i )$. 
The representation of a polytope as the intersection of half-spaces is commonly referred to as a half-space representation (H-representation) of the polytope. 
We will occasionally refer to polytopes in H-representation  as H-polytopes.

\begin{figure}[ht]
    \centering
     \subfloat[Net energy profiles and constraints]{\includegraphics[width=\columnwidth,trim= 0 6.3cm 0 0, clip]{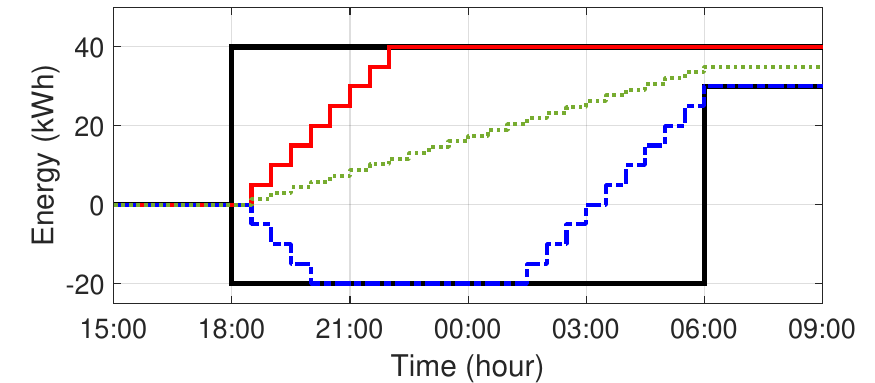}}\\
    \subfloat[Power profiles and constraints]{\includegraphics[width=\columnwidth,trim= 0 .3cm 0 6.3cm, clip]{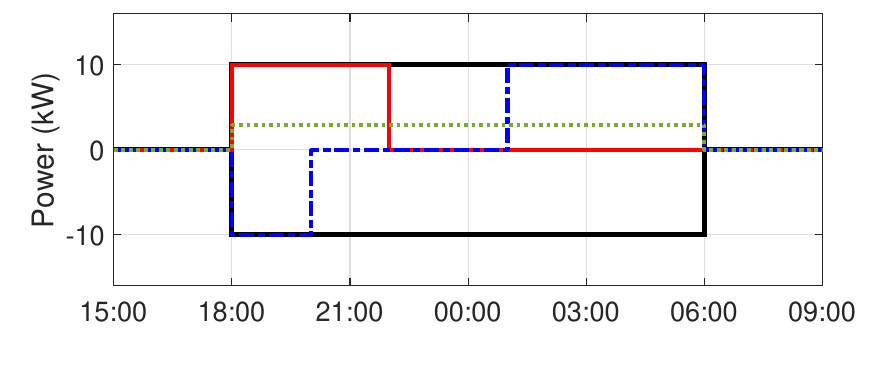}}
    \caption{Example of an individual EV flexibility set. The power and net energy profile constraints are depicted as solid black lines. Three different feasible power profiles and net-energy profiles are depicted. 
    In this example, we take $\delta=1/2$ hour, $T=24$, and associate the initial period $t=0$ with the 6:00-6:30 PM time interval. The EV charging parameters used in this example are: $a_i=0$  (6:00 PM arrival), $d_i=23$  (6:00 AM departure), $u_i^{\rm max}=-u_i^{\rm min} = 10$ kW, $x_i^{\rm max} = 60$ kWh, $x_i^{\rm init} = 20$ kWh, and $x_i^{\rm fin} = 50$ kWh.}
    \label{fig:single_EV}
\end{figure} 
\begin{exmp}[Constructing EV flexibility sets] \label{ex:ev_req} \rm  
    In this example, we show how to construct  individual flexibility sets from EV charging models used by real-world EV smart charging systems \cite{alexeenko2021achieving, lee2021adaptive}.
    Consider an EV $i \in \Ncal$ that is available to charge for a contiguous set of time periods between a plug-in time $a_i \in \Tcal$ and charging completion deadline $d_i \in \Tcal$.
    For all time periods $t \in \{a_i, \dots, d_i\}$,  EV $i$ can be charged or discharged at any power level between  given minimum and maximum rates, $u_i^{\rm min} \in \Rset$ and $u_i^{\rm max} \in \Rset$, respectively, where $ u_i^{\rm min} < u_i^{\rm max}$.  
    For all time periods  $t \notin \{a_i, \dots, d_i\}$, the charging rate is required to be zero. 
    Together, these charging profile constraints can be encoded as a generalized battery model \eqref{eq:indi_set} by specifying the charging profile limits $(\overline{u}_i, \, \underline{u}_i)$ according~to
    \begin{align}
    \label{eq:u_upper_ex} \overline{u}_i(t) & = u_i^{\rm max} \cdot \indi{ a_i \leq t \leq d_i}\\
    \label{eq:u_lower_ex} \underline{u}_i(t) & = u_i^{\rm min} \cdot \indi{ a_i \leq t \leq d_i}
    \end{align}
    for $t = 0, \dots, T-1$.  
    
    Additionally, EV $i$ is assumed to have a limited energy storage capacity $x_i^{\rm max} \in \Rset_+$, an initial state of charge $x_i^{\rm init} \in \Rset_+$ when plugging in to charge, and a desired final state of charge $x_i^{\rm fin} \in \Rset_+$ that must be satisfied by its charging completion deadline. 
    These requirements can also be expressed as constraints that are consistent with the family of generalized battery models \eqref{eq:indi_set} by specifying the net energy profile limits  $(\overline{x}_i, \, \underline{x}_i)$ according to
    \begin{align}
    \label{eq:x_upper_ex} \overline{x}_i(t)  & =  (x_i^{\rm max} -   x_i^{\rm init}) \cdot\indi{ t \geq a_i}\\
    \label{eq:x_lower_ex} \underline{x}_i(t) & =   x_i^{\rm fin} \cdot \indi{t > d_i} \  -  \ x_i^{\rm init} \cdot \indi{t \geq a_i}  
    \end{align}
    for $t = 1,  \dots, T$. 
    Collectively, Eqs. \eqref{eq:u_upper_ex}-\eqref{eq:x_lower_ex} completely define the individual flexibility set of EV $i$ as a generalized battery model.
    
    We provide an example of an individual EV flexibility set in Fig. \ref{fig:single_EV} using the power and energy profile limits specified in Eqs. \eqref{eq:u_upper_ex}-\eqref{eq:x_lower_ex}. 
    We remind the reader that the energy limits are bounds on the \emph{net energy} delivered to the EV. 
    For example, in Fig. \ref{fig:single_EV}, the EV has a battery capacity of $x_i^{\rm max} = 60 \text{ kWh}$ and an initial state-of-charge of $x_i^{\rm init} = 20 \text{ kWh}$. 
    As a result, the minimum and maximum net energy that can be supplied to the EV are $-20$ kWh and  $40$ kWh, respectively, as depicted in Fig. \ref{fig:single_EV}(a). 
\end{exmp}

\begin{rem}[Net-energy representation] \rm \label{rem:sparse}
An individual flexibility set can be equivalently represented in terms of the corresponding set of net-energy profiles, given by
\begin{align*}
    \Xset_i := L \Uset_i = \{ x \in\Rset^T \,|\, HL^{-1} x \le h_i \}.
\end{align*}
This alternative representation of the individual flexibility set may be advantageous from a computational perspective, as the left-hand side matrix $H L^{-1}$ is much sparser than the left-hand side matrix $H$ used in the power profile-based  representation given in \eqref{eq:power_rep}.
\end{rem}

\begin{rem}[Lossy charging dynamics] \rm \label{rem:nonideal}
The EV charging model \eqref{eq:dyn}   assumes lossless storage dynamics to ensure convexity of the individual flexibility sets. A lossy EV charging model with  time-varying energy leakage  and  conversion inefficiencies can be expressed as:
\begin{align*}
    x_i(t+1) = \zeta_i(t)  x_i(t) +  \delta\left(  \eta_i^{\rm in}(t) u_i(t)^+ +  \frac{1}{\eta_i^{\rm out}(t)} u_i(t)^-     \right), 
\end{align*}
where $u_i(t)^+ := \max(0, u_i(t))$  and $u_i(t)^- := \min(0, u_i(t))$,   and the  scalar coefficients $\zeta_i(t)$, $\eta_i^{\rm in}(t)$,  and $\eta_i^{\rm out}(t)$ are assumed to lie in the interval $(0, 1]$ for all $t \in \Tcal$. The above model reduces to the lossless storage model \eqref{eq:dyn} when  $\zeta_i(t) = \eta_i^{\rm in}(t) = \eta_i^{\rm out}(t) = 1$ for all $t \in \Tcal$.  The results  provided in this paper can be directly applied to the above class of lossy storage models for two important special cases: (i)  lossy storage dynamics with energy leakage (i.e., $\zeta_i(t) < 1 $) and no energy conversion inefficiencies (i.e., $\eta_i^{\rm in}(t) = \eta_i^{\rm out}(t) = 1$), and (ii) lossy storage dynamics with a one-way charging requirement, i.e., $u_i(t) \geq 0$ for all time periods $t \in \mathcal{T}$. 
In both cases, the resulting individual flexibility sets will be  compact, convex  polytopes in H-representation, enabling a direct application of the aggregation and disaggregation techniques developed in this paper. 
The treatment of EVs with general lossy storage dynamics is more challenging, because this results in \emph{nonconvex} individual flexibility sets. 
The extension of the techniques developed in this paper to account for general lossy charging dynamics  is left as a direction for future research.
\end{rem}

\subsection{Aggregate Flexibility Set}
The \emph{aggregate flexibility set} associated with a finite population of EVs can be expressed as a Minkowski sum of the individual  flexibility sets given by
\begin{align} \label{eq:agg_set}
\Uset := \sum_{i \in \Ncal} \Uset_i = \bigg \{ u \in \Rset^T \, \Big \vert \, u=\sum_{i\in\Ncal} u_i, \   u_i \in \Uset_i \bigg \}. 
\end{align}
Without loss of generality, we assume throughout the paper that the aggregate flexibility set  $\Uset$ is a full-dimensional polytope in $\Rset^T$. 

Note that it is NP-hard to compute the Minkowski sum of two H-polytopes \cite{tiwary2008hardness}. 
And while it is easy to compute the Minkowski sum of two convex polytopes in vertex representation  (V-representation), all known classes of algorithms that convert a polytope from H-representation to V-representation (vertex enumeration) and vice-versa (facet enumeration) exhibit worst-case complexities that are exponential in the polytope's number of dimensions. 
Since EV flexibility sets are typically provided as H-polytopes, calculating their aggregate flexibility set exactly is therefore computationally intractable in general.

\subsection{Approximating the Aggregate Flexibility Set}

Recognizing these challenges, our main objective in this paper is to devise computationally efficient methods to compute polyhedral inner approximations (i.e., subsets)  of the aggregate flexibility set $\Uset$. 
More precisely, we seek a polytope $\Pset$ that satisfies
\begin{align*}
\Pset \subseteq  \Uset.
\end{align*}
To facilitate the efficient calculation of inner  approximations of the aggregate flexibility set, we restrict our attention to approximating polytopes that are affine transformations of a given H-polytope  $\Uset_0 \subseteq \Rset^T$, i.e.,
\begin{align} \label{eq:trans}
\Pset = \overline{p} + P \Uset_0,
\end{align}
where  $\overline{p} \in \Rset^T$ and $P \in \Rset^{T \times T}$. 
Employing the same nomenclature as in \cite{sadraddini2019linear}, we refer to affine transformations of H-polytopes as AH-polytopes.  We will refer to the H-polytope $\Uset_0$ as the \emph{base set}, which is assumed to be fixed throughout the paper. 

Given  heterogeneous individual flexibility sets $\Uset_1,\dots,\Uset_N$, we are interested in computing a \emph{maximum-volume} AH-polytope  $\Pset = \overline{p} + P \Uset_0$ that is contained within the aggregate flexibility set $\Uset$ by solving the following optimal polytope containment problem
\begin{align} \label{eq:inner approx}
\text{maximize} \ \,  \vol{\Pset}  \ \, \text{subject to}  \  \  \Pset = \overline{p} + P \Uset_0 \subseteq \Uset,
\end{align}
with respect to the  optimization variables   $\overline{p} \in \Rset^T$ and $P \in \Rset^{T \times T}$. Here, $\vol{\cdot}$ denotes the $T$-dimensional volume function (a generalization of the usual volume measure in three dimensions to higher dimensions).  
We refer to feasible solutions to problem~\eqref{eq:inner approx}  as \emph{inner approximations}  of the aggregate flexibility set. 

The optimization problem \eqref{eq:inner approx} is challenging to solve for a variety of reasons. 
First, it is computationally intractable to exactly calculate the volume of a polytope in high dimensions \cite{dyer1988complexity}. 
Second, the polytope containment condition in problem \eqref{eq:inner approx} is also computationally intractable to verify in general \cite{sadraddini2019linear, tiwary2008hardness}.
Thus, instead of attempting to compute optimal solutions to problem \eqref{eq:inner approx},  we pursue a slightly less ambitious goal in this paper by seeking to \emph{conservatively approximate} problem \eqref{eq:inner approx} by a tractable convex program. 
The convex programs that we construct in Section \ref{sec:inner approximation results} are modestly-sized linear programs, which are guaranteed to generate valid inner  approximations of the aggregate flexibility set. 

\subsection{Choosing the Base Set}

The methods proposed in this paper rely on the a priori determination of a base set $\Uset_0$. 
Inspired by the approach taken in \cite{zhao2017geometric}, we restrict our attention to base sets of the~form
\begin{align} \label{eq:base set definition}
\Uset_0 := \left\{ u \in \Rset^T \, | \,   H u \leq h_0   \right\},
\end{align}
where the right-hand side vector $h_0 := (1/N) \sum_{i \in \Ncal} h_i$ is an average of the individual flexibility set parameters. 
This specific choice of  base set is intended to  approximate the collection of (potentially heterogeneous) individual flexibility sets in a balanced manner.

It was previously shown in  \cite[Proposition 1]{barot2017concise}  that a \emph{dilation} of this base set by a factor of $N$ results  in an outer approximation of the aggregate flexibility set, i.e., 
\begin{align} \label{eq:trivial_outer_approx}
    \Uset \subseteq N \Uset_0.
\end{align}
If the individual flexibility sets are identical, then this particular dilation of the given base set results in an exact expression for  the aggregate flexibility set, i.e., $\Uset = N \Uset_0$.  

It is important to note that the base set, being defined this way, belongs to the family of \emph{generalized battery models} defined in \eqref{eq:indi_set}. In certain applications, it may be necessary to restrict the family of allowable transformations in \eqref{eq:trans}  to those which are \emph{structure preserving}---i.e., transformations resulting in polytopes that are also generalized battery models. 
For example, independent system operators (ISOs) that manage wholesale electricity markets do not have the visibility or means  to effectively optimize the operation of individual EVs within a large aggregation. As a result, 
current market rules require  aggregators participating in wholesale markets to represent the collective capability of the resources under their control as a single representative resource  that can be dispatched by the ISO \cite{nyiso, caiso}.  
This leads us to the following definition of structure-preserving transformations.

\begin{defn}[Structure-preserving transformations] \label{def:struct} \rm An affine transformation $\Pset = \overline{p} + P\Uset_0$ is said to be \emph{structure preserving} if the resulting polytope can be expressed in H-representation as $\Pset = \{ u \in \Rset^T \, | \,  Hu \leq h'\}$ for some $h'\in \Rset^{4T}$.
\end{defn}
Importantly, a structure-preserving transformation of the proposed base set \eqref{eq:base set definition} produces a set that is within the generalized battery model class \eqref{eq:indi_set}, which is compatible with the resource aggregation representations currently used in wholesale electricity markets. 
In Section \ref{sec:structure preserving innner}, we provide a linear programming-based approach to compute a structure-preserving transformation of the base set that is guaranteed to be a subset of the aggregate flexibility set.

Also, note that a transformation given by a translation and dilation of the base set is structure preserving. To see this, let  $\Pset = \overline{p} + \alpha \Uset_0$, where  $\alpha >0 $. Under this transformation, it holds  that  $\Pset = \{ u \in \Rset^T \, | \,  Hu \leq \alpha h_0 + H\overline{p} \}.$
It should also be noted that while we have adopted a specific choice of base set in Eq. \eqref{eq:base set definition}, all of the following results provided in this paper hold for any choice of base set that is an H-polytope. 

%% file: Inner_Approximation.tex
\section{Inner Approximation Methods} \label{sec:inner approximation results}
\begin{figure*}[htb!]
    \centering 
    \includegraphics[width=0.9\linewidth]{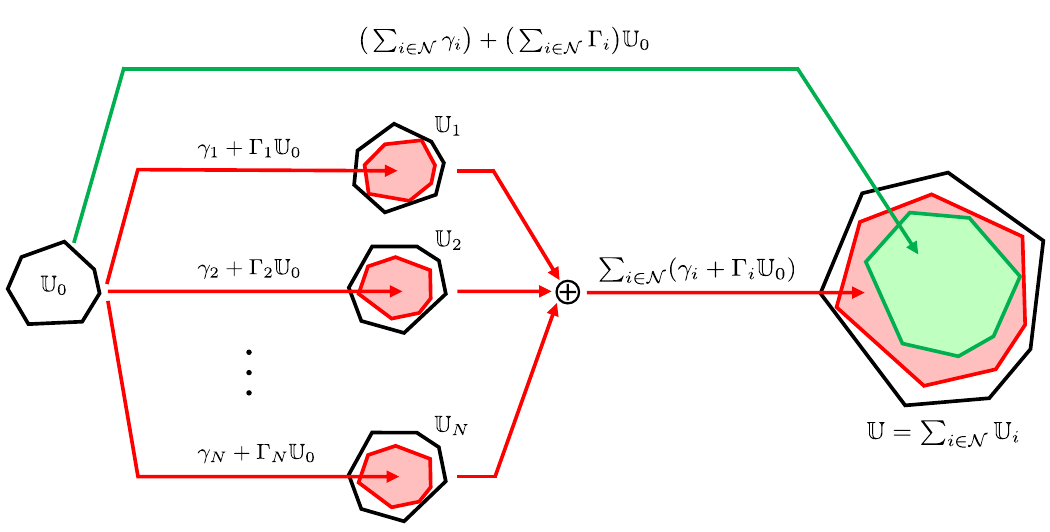}
    \caption{Illustration of the inner approximation method proposed in this paper. Depicted are: the base set $\Uset_0$; the affine transformations $\gamma_i+\Gamma_i\Uset_0$ (red) that inner approximate the individual flexibility sets $\Uset_i$  for $i=1,\dots,N$; and the corresponding affine transformation $(\sum_{i\in\Ncal}\gamma_i)+(\sum_{i\in\Ncal}\Gamma_i)\Uset_0$ (green) that inner approximates the aggregate flexibility~set~$\Uset$.}
    \label{fig:inner_approx}
\end{figure*}

In this section, we derive a conservative linear programming (LP) approximation of the optimal inner polytope containment problem~\eqref{eq:inner approx}.
We start by constructing an inner approximation to each individual flexibility set given by: 
\begin{align} \label{eq:ind_con}
    \gamma_i+\Gamma_i\Uset_0 \subseteq \Uset_i, \quad i = 1, \dots, N.
\end{align}
Here, $\gamma_i\in\Rset^T$ and $\Gamma_i\in\Rset^{T \times T}$ ($i = 1, \dots, N$)  are optimization variables that will be selected to ensure that each AH-polytope $\gamma_i  + \Gamma_i \Uset_0$ closely approximates its corresponding flexibility set $\Uset_i$, while satisfying the individual set containment conditions \eqref{eq:ind_con}, as depicted by the red arrows in the left-hand side of Fig. \ref{fig:inner_approx}. It follows from  \eqref{eq:ind_con} that the Minkowski sum of the resulting AH-polytopes  is  an inner approximation of  the aggregate flexibility set, i.e., 
\begin{align} \label{eq:mink_sum}
    \sum_{i \in \Ncal} \gamma_i+\Gamma_i\Uset_0  \subseteq \sum_{i\in \Ncal} \Uset_i,
\end{align}
as illustrated by the red arrow in the right-hand side of Fig.~\ref{fig:inner_approx}.
However, the Minkowski sum  $\sum_{i \in \Ncal} \gamma_i+\Gamma_i\Uset_0$ is still intractable to compute. We avoid this difficulty by summing the elements of the individual transformations to create the mapping depicted by the green arrow in the upper half of Fig.  \ref{fig:inner_approx}. This provides an inner approximation of the aggregate flexibility set  due to the following property\footnote{To see why the inclusion \eqref{eq:sum_ah} is true, note that any  element  $u \in \big(\sum_{i\in\Ncal} \gamma_i\big) + \big(\sum_{i\in\Ncal} \Gamma_i\big) \Uset_0 $ can be expressed as $u = \sum_{i \in \Ncal} (\gamma_i + \Gamma_i u_0)$ for some $u_0 \in \Uset_0$. Since the element  $\gamma_i + \Gamma_i u_0$  belongs to the set $\gamma_i + \Gamma_i \Uset_0$ for  each $i \in \Ncal$, it follows that $u \in \sum_{i \in \Ncal}  (\gamma_i + \Gamma_i \Uset_0$). }
\begin{align}\label{eq:sum_ah}
    \Big(\sum_{i\in\Ncal} \gamma_i\Big) + \Big(\sum_{i\in\Ncal} \Gamma_i\Big) \Uset_0 \subseteq  \sum_{i\in\Ncal} \gamma_i + \Gamma_i \Uset_0. 
\end{align}
Setting $\overline{p} =\sum_{i\in\Ncal} \gamma_i$ and $P=\sum_{i\in\Ncal} \Gamma_i$ yields an inner approximation to the aggregate flexibility set given by  $\Pset = \overline{p} + P \Uset_0 \subseteq \Uset$, which follows from  inclusions \eqref{eq:mink_sum} and \eqref{eq:sum_ah}.

As a key building block in the construction of a convex approximation to problem \eqref{eq:inner approx}, we provide a set of linear constraints that are necessary and sufficient  for the containment of an AH-polytope within an H-polytope. 

\begin{lem}[AH-polytope in H-polytope] \label{lem:ah-polytope containment} \rm Let $\Xset = \{ x \in \Rset^{n_x} \, | \, H_x x \leq h_x\}$ and $\Yset = \{ y \in \Rset^{n_y} \, | \, H_y y \leq h_y\}$,   where $H_x \in \Rset^{m_x \times n_x}$, $H_y \in \Rset^{m_y \times n_y}$, and $\Xset$ is assumed to be nonempty. Given a vector $\gamma \in \Rset^{n_y} $ and matrix $\Gamma \in \Rset^{n_y \times n_x}$, it holds that $\gamma + \Gamma \Xset \subseteq \Yset$  if and only if there exists a matrix $\Lambda \in \Rset^{m_y \times m_x}$ such that
\begin{align}
\label{eq:inner cond 1} & \Lambda  \geq 0,  \\
\label{eq:inner cond 2}  &\Lambda  H_x = H_y\Gamma ,    \\
 \label{eq:inner cond 3}   &\Lambda h_x  \leq h_y -  H_y \gamma.
\end{align}
\end{lem}

Lemma \ref{lem:ah-polytope containment} is a known result in the literature \cite{kellner2015containment, rakovic2007optimized, sadraddini2019linear}. 
It follows from standard duality results in convex analysis, and can be interpreted as a variant of Farkas' Lemma.  
To keep the paper self contained, we include a simple proof that uses the strong duality property of linear programs.

\begin{proof}  First, notice that the set inclusion $\gamma + \Gamma \Xset \subseteq \Yset$  holds if and only if the AH-polytope  $ \gamma + \Gamma \Xset$ is contained in each half-space defining the H-polytope $\Yset$, i.e.,
\begin{align} \label{eq:proof halfspace}
 \underset{x \in \Xset}{\sup}\ H_{y,j} ^\top (\gamma + \Gamma x) \leq h_{y,j}, \quad  j = 1, \dots, m_y, 
\end{align}
where $H_{y,j} ^\top$ denotes the $j$-th row of $H_y$  and $h_{y,j}$ denotes the $j$-th element of $h_y$.  For $j = 1, \dots, m_y$, \eqref{eq:proof halfspace} is equivalent to
\begin{align*}
 & h_{y,j}  - H_{y,j}^\top \gamma \geq    \underset{x \in \Rset^{n_x} }{\sup}  \left\{ H_{y,j} ^\top \Gamma x \, | \, H_x x \leq h_x    \right\}, \\
\Leftrightarrow \, & h_{y,j}  - H_{y,j}^\top \gamma \geq    \underset{\lambda_j \in \Rset^{m_x}_+ }{\inf} \left\{  \lambda_j^\top h_x  \, | \,   H_x^\top \lambda_j = \Gamma^\top H_{y,j} \right\}, \\
\Leftrightarrow  \, &    \begin{aligned}\exists \, \lambda_j \in \Rset^{m_x}_+ \text{ s.t. } & \lambda_j^\top h_x \leq  h_{y,j}  - H_{y,j}^\top \gamma, \,   H_x^\top \lambda_j = \Gamma^\top H_{y,j}.\end{aligned}
\end{align*}
The equivalence in the second line follows from the strong duality of linear programs, as the primal problem 
has  a nonempty feasible set. By defining $\Lambda := (\lambda_1^\top, \dots, \lambda_{m_y}^\top)$, the conditions in the third line can be shown to be equivalent to conditions \eqref{eq:inner cond 1}, \eqref{eq:inner cond 2}, and  \eqref{eq:inner cond 3}, proving the desired result.
\end{proof}

Lemma \ref{lem:ah-polytope containment} can be used to linearly encode the individual set containment conditions \eqref{eq:ind_con}. 
Using this linear reformulation 
in combination with property \eqref{eq:sum_ah}, we derive a 
set of sufficient conditions  for the set containment constraint $\overline{p} + P \Uset_0 \subseteq \Uset$.

\begingroup
\allowdisplaybreaks
\begin{thm}[AH-polytope in Sum of H-polytopes]  \label{thm:main inner containment} \rm It holds that $\overline{p} + P \Uset_0 \subseteq \Uset$ if there exist  $\gamma_i \in \Rset^{T}$,  $\Gamma_i \in \Rset^{T \times T}$, and $\Lambda_i \in \Rset^{4T \times 4T}$ for $i=1,\dots, N$ such that 
\begin{align}
 \label{eq:suff inner 1} &\hspace{0in} [\, \overline{p}, \, P \,] = \sum\nolimits_{i=1}^N  [\,\gamma_i, \, \Gamma_i \,], \\
 \label{eq:suff inner 2} & \Lambda_i  \geq 0, \quad i = 1, \dots, N, \\
 \label{eq:suff inner 3}  &\Lambda_i  H = H \Gamma_i,   \quad i = 1,\dots, N,  \\
  \label{eq:suff inner 4}   &\Lambda_i h_0 \leq h_i  - H \gamma_i, \quad i = 1,\dots, N.
\end{align}
\end{thm}
\endgroup

The sufficient conditions  provided in Theorem \ref{thm:main inner containment} are linear with respect to the variables   $\overline{p}$, $P$,  $\gamma_i$,  $\Gamma_i$, and  $\Lambda_i$ $(i=1, \dots, N)$. As a result, the set containment constraint $\overline{p} + P \Uset_0 \subseteq \Uset$ can be conservatively approximated by a finite set of linear constraints in these variables, where the resulting number of decision variables and constraints scales polynomially with the size of the input data. 

\begin{proof} It follows from Lemma \ref{lem:ah-polytope containment} that, for each $i \in \Ncal$,  conditions \eqref{eq:suff inner 2}, \eqref{eq:suff inner 3}, and \eqref{eq:suff inner 4} are necessary and sufficient for the set inclusion $\gamma_i + \Gamma_i\Uset_0 \subseteq \Uset_i$. 
The desired result then follows from \eqref{eq:mink_sum}, \eqref{eq:sum_ah} and \eqref{eq:suff inner 1}  which together imply that
$\overline{p} + P \Uset_0 \subseteq \sum_{i \in \Ncal}  \gamma_i + \Gamma_i \Uset_0 \subseteq  \sum_{i \in \Ncal} \Uset_i =   \Uset.$
\end{proof}

With the conservative linear approximation of the containment constraint $\overline{p} + P \Uset_0 \subseteq \Uset$ provided by Theorem  \ref{thm:main inner containment} in hand,  we now turn to the problem of approximating the optimal inner polytope containment problem \eqref{eq:inner approx} by a linear program under structure-preserving affine transformations in Section \ref{sec:structure preserving innner}, and general affine transformations in Section~\ref{sec:general affine inner}.
The former will yield inner approximations of the aggregate flexibility set that have a battery representation and can be used to participate in the wholesale electricity market. The latter will yield inner approximations that are  potentially more accurate and faster to compute, but may lack a battery representation.

\begin{rem}[Heterogeneity in EV charging dynamics] \rm
    The conditions in Theorem \ref{thm:main inner containment} can be modified to accommodate  individual flexibility sets with different left-hand side matrices $H$. This allows individual EVs to have different charging dynamics parameters, such as different energy leakage coefficients and energy conversion inefficiencies, as described in Remark \ref{rem:nonideal}.
\end{rem}

\subsection{Structure-Preserving Transformations} \label{sec:structure preserving innner}
We first consider \emph{structure-preserving transformations} obtained by a translation and positive scaling of the base set: $$\Pset = \overline{p} + \alpha \Uset_0,$$
where $\alpha > 0$ denotes the scaling factor. 
Given this restriction on the family of allowable transformations, the volume of the inner approximating polytope $\Pset$ can be expressed as $\vol{\Pset} = |\det(\alpha I_T)| \vol{\Uset_0} = \alpha^T \vol{\Uset_0}. $\footnote{This  follows from the identity  $\vol{A\Xset} =  |\det(A)|\vol{\Xset}$, which gives the volume of a set $\Xset \subseteq \R^T$ under a linear transformation $A \in \Rset^{T \times T}$.}
Since the base set $\Uset_0$ is assumed to be fixed throughout the paper, maximizing the volume of $\Pset$ is equivalent to maximizing the scaling factor $\alpha$.  Using this fact in combination with the sufficient containment conditions provided by Theorem \ref{thm:main inner containment}, we arrive at the following conservative approximation of the original optimal inner polytope containment problem \eqref{eq:inner approx}:
\begin{align}
\nonumber \textrm{maximize}\ \, \quad &  \alpha \\
\nonumber \textrm{subject to } \quad &      [\, \overline{p}, \, \alpha I_T \,] = \sum\nolimits_{i=1}^N  [\, \gamma_i, \, \Gamma_i \,], \\
\nonumber & \alpha > 0, \\
\label{eq:LP inner approx}   & \Lambda_i  \geq 0, \quad i = 1, \dots, N, \\
\nonumber&\Lambda_i  H = H \Gamma_i,   \quad i = 1,\dots, N,  \\
\nonumber &\Lambda_i h_0 \leq h_i  - H \gamma_i, \quad i = 1,\dots, N.
\end{align}
Problem \eqref{eq:LP inner approx} is a linear program (LP) in the decision variables  $\overline{p}$, $\alpha$, $\gamma_i$, $\Gamma_i$, and $\Lambda_i$ $(i = 1, \dots, N)$. 

\begin{rem}[Comparison to the method in \cite{zhao2017geometric}] \rm \label{rem:struc_preserv}
The approach proposed by Zhao \emph{et al.} \cite{zhao2017geometric} entails finding a maximal inner approximation to each individual flexibility set using homothetic transformations of the given base set by solving: 
\begin{align} \label{eq:zhao_inner}
    \text{maximize} \ \,  \alpha_i  \ \, \text{subject to}  \  \  \gamma_i + \alpha_i \Uset_0 \subseteq \Uset_i 
\end{align}
for every EV $i \in \Ncal$.
The decision variables are the translation $\gamma_i \in \Rset$ and scaling $\alpha_i \in \Rset$ parameters for each homothet.
Summing the resulting homothets  yields a structure-preserving  inner approximation to the aggregate flexibility set given by $\left( \sum_{i\in\Ncal} \gamma_i \right) + \left(\sum_{i\in\Ncal} \alpha_i \right)\Uset_0$. 
The restriction to homothetic transformations in \eqref{eq:zhao_inner}  (a special case of the affine transformations proposed in this paper) can result in overly conservative approximations if the individual flexibility sets differ significantly in shape or dimension from the  base set  $\Uset_0$.  
In particular, when an individual flexibility set has lower dimension than the base set (i.e., $\dim \Uset_i < \dim \Uset_0$), the only feasible homothetic approximations are singleton sets (with a single charging profile) induced by a zero scaling factor $\alpha_i  = 0$.

Our proposed method addresses these shortcomings by optimizing over general affine transformations of the base set given by $\gamma_i + \Gamma_i \Uset_0$ ($i =1, \dots, N$), requiring only that the resulting AH-polytope $(\sum_{i \in \Ncal} \gamma_i) + (\sum_{i \in \Ncal} \Gamma_i)\Uset_0$  be structure preserving by enforcing the constraint $\sum_{i \in \Ncal} \Gamma_i = \alpha I_T$.
This enlargement of the set of structure-preserving approximations can significantly improve approximation accuracy when the individual flexibility sets are heterogeneous in shape. 
In Fig.~\ref{fig:2d-example}, we provide a two-dimensional example that  highlights some of the advantages of the approximation methods proposed in this paper in comparison to the homothet-based approximation method of Zhao \emph{et al.}  \cite{zhao2017geometric}.
By using affine transformations of the base set to approximate each of the individual flexibility sets, our general affine and structure-preserving approximation methods yield inner approximations to the aggregate flexibility set  that are supersets of the homothet-based inner approximation for this particular example.  
It is also important to note that, in general, our approximations are neither supersets nor subsets of the zonotope-based approximation of  M{\"u}ller \emph{et al.} \cite{muller2017aggregation}. We draw a more extensive comparison between these methods using realistic case studies in Section \ref{sec:experiments}.
\end{rem}

We note that although our structure-preserving inner approximation method requires the solution of an LP \eqref{eq:LP inner approx} with  $O(NT^2)$ decision variables, this LP has favorable  sparsity structure that can be exploited by decomposition methods to improve solve times. In particular, the LP \eqref{eq:LP inner approx}  consists of $N$ individual polytope containment subproblems that are coupled only through the `structure-preserving' constraint $\sum_{i \in \Ncal} \Gamma_i = \alpha I_T$ and the objective function.
This \emph{block-angular sparsity structure} is well-suited for the application of decomposition methods such as the \emph{Dantzig-Wolfe decomposition} \cite{dantzig1960decomposition}.
 
\begin{figure*}
    \centering 
    \includegraphics[width=2.0\columnwidth, trim=2.75cm 9.5cm 0 1.2cm, clip]{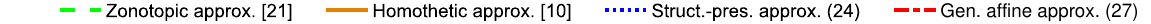}\\
    \subfloat[Individual flexibility set, $\Uset_1 \subseteq \Rset^2$]{\includegraphics[width=0.3\linewidth, trim=0.5cm 0.2cm 1cm 0.4cm, clip]{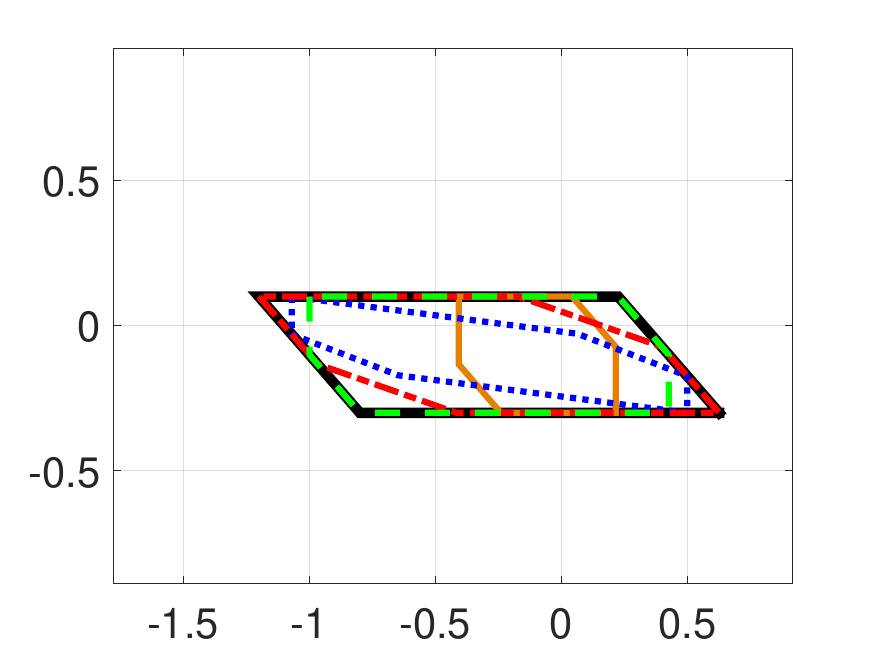}\label{fig:u1}}\quad
    \subfloat[Individual flexibility set, $\Uset_2 \subseteq \Rset^2$]{\includegraphics[width=0.3\linewidth, trim=0.5cm 0.2cm 1cm 0.4cm, clip]{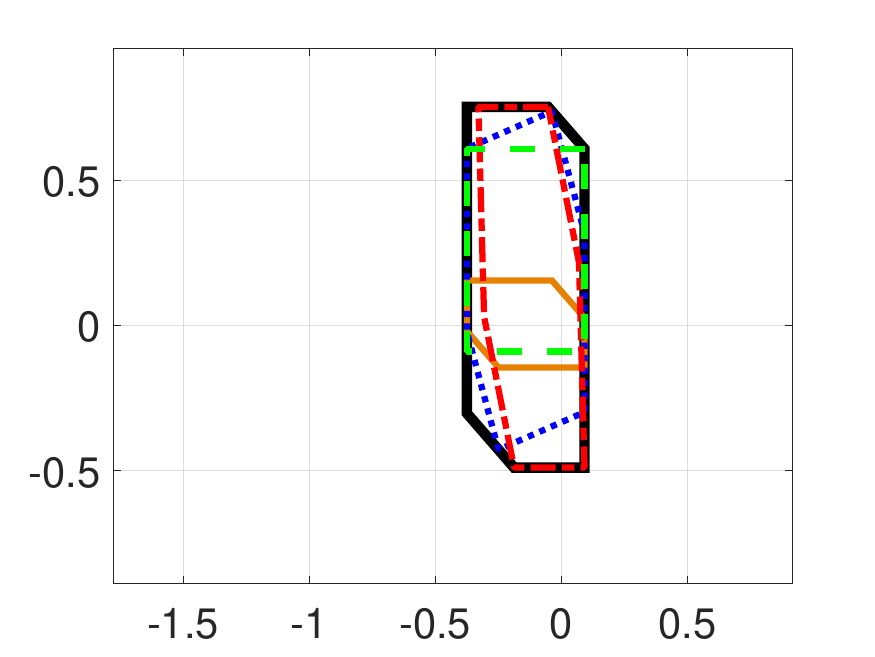}\label{fig:u2}}\quad
    \subfloat[Minkowski sum, $\Uset= \Uset_1 + \Uset_2$]{\includegraphics[width=0.3\linewidth, trim=0.5cm 0.2cm 1cm 0.4cm, clip]{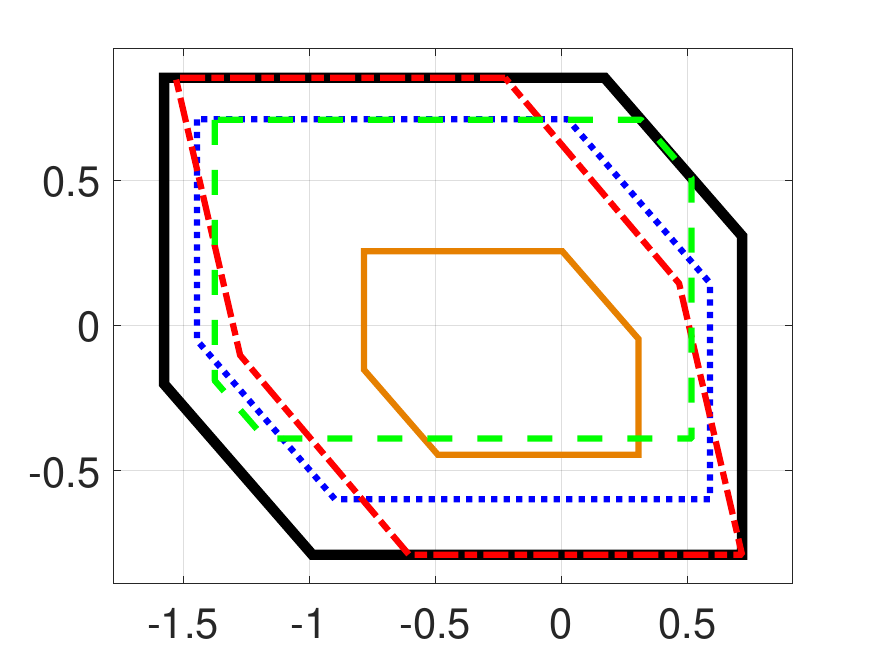}\label{fig:usum}}
    \caption{Comparison of inner approximation methods. (a), (b) Two individual flexibility sets $\Uset_i = \{x \mid Hx \leq h_i\}$ ($i= 1,2$) with randomly sampled right-hand side vectors and  (c) their sum $\Uset = \Uset_1 + \Uset_2$ are depicted as black solid lines. }
    \label{fig:2d-example}
\end{figure*}

\subsection{General Affine Transformations}  \label{sec:general affine inner}
We now show how to conservatively approximate the original optimal inner polytope containment problem \eqref{eq:inner approx} by a LP when considering more general affine transformations of the base set.  First, note that the volume of the transformation $\Pset = \overline{p} + P \Uset_0$ is given by $\vol{\Pset} = | \det (P)| \vol{\Uset_0}$.  Hence, maximizing $\vol{\Pset}$ is equivalent to maximizing $|\det (P)|$. As this function is nonconcave over the set of real square matrices, we linearize $|\det (P)|$ using a first-order Taylor expansion about the identity matrix to obtain:
\begin{align} \label{eq:linearization of det}
 \vol{\Pset} \propto |\det (P)| \approx  \trace{P} + \text{constant}. 
\end{align} 
Here, we have used the fact that, for nonsingular matrices $P$, the gradient of  $|\det (P)|$ with respect to $P$ is given by $\nabla_P |\det (P)| = |\det(P)|(P^{-1})^\top$.

Employing the  sufficient containment condition provided by Theorem \ref{thm:main inner containment} in combination  with the linear approximation of the volume objective function  in \eqref{eq:linearization of det}  leads to the following conservative approximation of the original  optimal inner polytope containment problem \eqref{eq:inner approx}:
\begin{align}
\nonumber \textrm{maximize}\ \, \quad &  \trace{P} \\
\nonumber \textrm{subject to } \quad &      [\,\overline{p}, \, P\,] = \sum\nolimits_{i=1}^N  [\,\gamma_i, \, \Gamma_i\,], \\
\label{eq:LP inner approx 2}   & \Lambda_i  \geq 0, \quad i = 1, \dots, N, \\
\nonumber&\Lambda_i  H = H \Gamma_i,   \quad i = 1,\dots, N,  \\
\nonumber &\Lambda_i h_0 \leq h_i  - H \gamma_i, \quad i = 1,\dots, N.
\end{align}
Problem \eqref{eq:LP inner approx 2} is a LP in the decision variables  $\overline{p}$, $P$, $\gamma_i$, $\Gamma_i$, and $\Lambda_i$ $(i = 1, \dots, N)$.  

Note that problem \eqref{eq:LP inner approx 2} reduces to the structure-preserving LP \eqref{eq:LP inner approx} under the additional  restriction that $P = \alpha I_T$ and $\alpha > 0$. In Section \ref{sec:experiments}, we conduct numerical experiments illustrating the improvement in approximation accuracy achievable by optimizing over the more general family of affine transformations encoded in problem \eqref{eq:LP inner approx 2}. This improvement in approximation accuracy is also illustrated in the two-dimensional example provided in Fig. \ref{fig:2d-example}, where the inner approximation produced by the LP in \eqref{eq:LP inner approx 2} is depicted in red.           

\begin{rem}[Decomposition] \rm 
It is important to note that, unlike the structure-preserving LP  \eqref{eq:LP inner approx}, problem \eqref{eq:LP inner approx 2} possesses  \emph{block-separable structure} in the variables $\gamma_i$, $\Gamma_i$, and $\Lambda_i$ $(i = 1, \dots, N)$. This structure can be exploited to decompose problem  \eqref{eq:LP inner approx 2} into $N$ separate LPs given by
\begin{align}
\nonumber \textrm{maximize}\ \, \quad &  \trace{\Gamma_i} \\
 \textrm{subject to } \quad &     
\label{eq:decomp 1}    \Lambda_i  \geq 0, \\
\nonumber&\Lambda_i  H = H \Gamma_i,   \\
\nonumber &\Lambda_i h_0 \leq h_i  - H \gamma_i,
\end{align}
for $i=1, \dots, N$. The decomposed LPs  are equivalent to
\begin{align} \label{eq:decomp 2}
    \textrm{maximize} \ \,  \trace{\Gamma_i} \ \,  \textrm{subject to} \ \ \gamma_i + \Gamma_i \Uset_0 \subseteq \Uset_i
\end{align}
for $i = 1, \dots, N$. The equivalence between problems \eqref{eq:decomp 1} and \eqref{eq:decomp 2} follows from Lemma \ref{lem:ah-polytope containment}.
Crucially, these LPs can be be solved sequentially or in parallel, producing  an optimal solution to the original problem \eqref{eq:LP inner approx 2} via the reconstruction $[\,\overline{p}, \, P\,] = \sum\nolimits_{i=1}^N  [\,\gamma_i, \, \Gamma_i\,]$.
\end{rem}

We conclude this section by reminding the reader that the proposed LP approximation \eqref{eq:LP inner approx 2} is based on a linearization of the volume objective function about the identity matrix. It may be possible to improve upon the quality of solutions generated by this approximation by using an iterative linearization-maximization method \cite{ortega2000iterative} to locally maximize the volume objective function. As another approach, one can use an objective function that more closely mirrors the downstream task where the approximation to the aggregate flexibility will be utilized.  We leave this as a direction for future work.

\section{Disaggregation Method}  \label{sec:disaggregation}
To implement a given aggregate power profile $u \in \Uset$, the aggregator must \emph{disaggregate} that profile into a collection of individual power profiles that can be feasibly executed by each EV in the given population. 
This corresponds to finding $N$ profiles $u_i \in \Uset_i$ for $i = 1, \dots, N$ such that  $u = \sum_{i \in \Ncal} u_i$. 
This can be achieved by solving a linear feasibility problem whose size grows with the number of EVs.

Alternatively,  using the class of inner approximations provided in this paper, one can avoid having to solve a LP for disaggregation. 
Specifically, the computation of an inner approximation according to the conditions in Theorem~\ref{thm:main inner containment} yields, as a byproduct, an affine mapping that can transform any feasible point in the inner approximation into a collection of individually feasible points.   

To better understand this approach to disaggregation, let $\Pset = \overline{p} + P \Uset_0 \subseteq \Uset$ denote an inner approximation to the aggregate flexibility set satisfying  conditions \eqref{eq:suff inner 1}-\eqref{eq:suff inner 4} in Theorem \ref{thm:main inner containment}, and let $u \in \Pset$  denote an arbitrary point in this set.  It follows from \eqref{eq:suff inner 1} that there exists a point $u_0 \in \Uset_0$  such that the given point $u$ can be expressed as
\begin{align}
    \label{eq:disagg} u = \overline{p} + P u_0 =  \sum_{i \in \Ncal} \gamma_i + \Gamma_i u_0.
\end{align}
Additionally, it follows from conditions \eqref{eq:suff inner 2}-\eqref{eq:suff inner 4} that $\gamma_i + \Gamma_i\Uset_0 \subseteq \Uset_i$ for all $i \in \Ncal$ (a direct consequence of Lemma~\ref{lem:ah-polytope containment}). This implies that  $\gamma_i + \Gamma_i u_0 \in \Uset_i $ for all $i \in \Ncal$. Thus, the given point $u \in \Pset$ can be disaggregated into a collection of individually feasible points given by
\begin{align} \label{eq:disagg elements}
  u_i := \gamma_i + \Gamma_i u_0
\end{align}
for $ i = 1, \dots, N$. If the matrix $P$ is also invertible, then the disaggregated power profiles in \eqref{eq:disagg elements} can be expressed as  explicit functions of the aggregate power profile $u$ as follows: 
\begin{align} \label{eq:disagg elem 2}
  u_i := \gamma_i + \Gamma_i P^{-1}(u - \overline{p})
\end{align}
for $i = 1, \dots, N$. 

%% file: Experiments.tex
\section{Case Studies} \label{sec:experiments}
We compare the inner approximation methods proposed in this paper with competing methods in the literature by examining two practical applications of bidirectional EV charging: (i) peak power minimization and (ii) electricity cost minimization. 

The \emph{peak power minimization problem} is defined as 
\begin{align}\label{eq:power_min}
    \text{minimize} \ \, \| u + \ell \|_\infty  \ \, \text{subject to}  \  \, u  \in \Uset,
\end{align}
where $\Uset \subseteq \Rset^T$ denotes the aggregate flexibility set associated with a given set of participating EVs, and $\ell \in \Rset^T$  (kW) denotes the aggregate load profile associated with a given set of households. We refer to $u + \ell$ as the net-load profile. When there is a positive (negative) net load at time $t$, it means that energy is being drawn from (fed into) the grid during that specific time period. 

The \emph{electricity cost minimization problem} is defined as 
\begin{align}\label{eq:cost_min}
    \text{minimize} \ \,  (p^\top u) \delta  \ \, \text{subject to}  \  \, u  \in \Uset,
\end{align}
where $p \in \Rset^T$ (\$/kWh) denotes a given sequence of energy prices. In addition to fulfilling the charging requirements of participating EVs, an aggregator equipped with bidirectional charging capabilities can also exploit price arbitrage opportunities by buying energy from the grid when prices are low and selling it back when prices are higher.

In both of these problem formulations, we have made a number of simplifying assumptions. First, we assume that the aggregator has perfect knowledge of the aggregate household load profile $\ell$ and the price profile $p$ at the outset (but these  could be forecasted in practice). Second, we assume that the energy prices are unaffected by the  buy/sell actions of the aggregator. This is a reasonable assumption for aggregators  that are not large enough to exert market power.

We will use problems \eqref{eq:power_min}  and \eqref{eq:cost_min} as the basis for numerical experiments designed to compare the effectiveness of the structure-preserving and general affine approximation methods proposed in this paper, with the homothet-based approximation methods of Zhao \emph{et al.} \cite{zhao2017geometric}, and the zonotope-based approximation methods of M{\"u}ller \emph{et al.} \cite{muller2017aggregation}.
For a more comprehensive survey and comparison of inner approximation methods for aggregate flexibility sets, we refer the reader to  \cite{ozturk2022aggregation}.

\begin{table}[t!]
    \centering
    \begingroup
    \setlength{\tabcolsep}{6pt} 
    \renewcommand{\arraystretch}{1.4} 
    \begin{tabular}{clc}
        \hline
        \hline
        \textbf{Param.} & \textbf{Description} & \textbf{Value/Range}  \\
        \hline   
        $\delta$ & Time period length  &   1 hr \\
        $T$ & Time horizon & 18 \\
        $N$ & Number of EVs & 25 \\
        $a_i$ & Plug-in time period &   0 ($3\text{ PM}$ arrival) \\
        $d_i$ &  Deadline time period & 17 ($9 \text{ AM}$ departure) \\
        $x_i^{\rm max}$ & Battery capacity  & $[25,50]$ kWh \\
        $u_i^{\rm max}$ & Max charging rate  & $[3,10]$ kW \\
        $u_i^{\rm min}$ & Min charging rate  & $[-10,-3]$ kW \\
        $x_i^{\rm init}$ & Initial state-of-charge  & $[0, 0.4 x_i^{\rm max} ] $ kWh  \\
        $x_i^{\rm fin}$ & Final state-of-charge  & $[0.6 x_i^{\rm max} , x_i^{\rm max}] $ kWh  \\ \hline \hline
    \end{tabular}
    \endgroup
    \caption{Summary of EV charging parameters used in~experiments.  The parameters are either fixed at~the specified value or uniformly distributed random variables over the specified interval.  We associate the initial time period $t=0$ with the 3:00-4:00 PM time interval.}
    \label{tab:data_gen}
\end{table}

\subsection{Data Description}
The experiments are carried out using historical load and energy price data. The load data,  obtained from the Pecan Street Dataport \cite{street2016dataport}, consists of electricity consumption profiles (excluding solar power production) for 25 individual households in Tompkins County, New York. The load data spans a six-month period between May 1, 2019 and October 31, 2019. We add up the individual household load profiles to obtain an aggregate load profile $\ell$ for each day in the given data set (184 days in total). For the price data, we utilized historical day-ahead (DA) energy prices from the NYISO Central Zone, making sure to align the dates and times of the sampled DA energy prices with the given load data.

The EV charging data used in our numerical experiments are simulated to reflect typical overnight  charging requirements.
Table \ref{tab:data_gen} summarizes the EV charging parameters used along with their specific values or the intervals from which they are randomly sampled.
All of the random variables are assumed to be mutually independent.
Using the simulated EV arrival/departure times, energy requirements, and charging constraints, we construct individual flexibility sets using the procedure outlined in Example \ref{ex:ev_req}.

As discussed in Remark \ref{rem:struc_preserv}, the homothet-based approximation method of Zhao \emph{et al.} \cite{zhao2017geometric} requires that $\dim \Uset_i \geq \dim \Uset_0$ for all $i \in \Ncal$. 
In the context of the EV charging model considered in our experiments, this corresponds to a requirement that all EVs have identical arrival and departure times. 
Thus, although the approximation methods proposed in this paper can accommodate EVs with heterogeneous arrival and departure times, we assume that $a_i = 0$ and $d_i = T-1$ for all $i \in \Ncal$ to facilitate a comparison with \cite{zhao2017geometric}.

\subsection{Experiments Description}

For each day in the six-month period under consideration, we randomly sample a finite collection of individual EV flexibility sets according to the parameters specified in Table~\ref{tab:data_gen}. Using the sampled individual flexibility sets,  we compute an inner approximation of the corresponding aggregate flexibility set using each of the following methods:  
\begin{enumerate}[(i)]
    \item General affine approximation [this paper],
    \item Structure-preserving approximation [this paper],
    \item Homothet-based approximation \cite{zhao2017geometric},
    \item Zonotope-based approximation \cite{muller2017aggregation}.\footnote{The zonotope-based approximations are computed using MATLAB code provided by M{\"u}ller \emph{et al.} \cite{muller2017aggregation}.}
\end{enumerate}
Using each of the resulting inner approximations, we solve the peak power minimization problem \eqref{eq:power_min} and the electricity cost minimization problem \eqref{eq:cost_min} (replacing the true aggregate flexibility set $\Uset$ with the corresponding inner approximation). As a benchmark for comparison, we also solve the optimization problems  \eqref{eq:power_min} and \eqref{eq:cost_min} using the true aggregate flexibility set.

To assess the performance of each approximation method (i)-(iv), we compute the differences between the suboptimal values obtained by solving the inner approximations to problems \eqref{eq:power_min} and \eqref{eq:cost_min} and the optimal values obtained using the true aggregate flexibility set. By repeating these calculations for every day in the given data set, we obtain suboptimality gap distributions for each approximation method (i)-(iv), as applied to both problems \eqref{eq:power_min} and \eqref{eq:cost_min}. The resulting suboptimality gap distributions are reported in Figure \ref{fig:subopt_gap}.

\subsection{Results and Discussion}

\begin{figure}
    \centering
    \subfloat[Peak power suboptimality gap distributions]{\includegraphics[width=0.85\columnwidth]{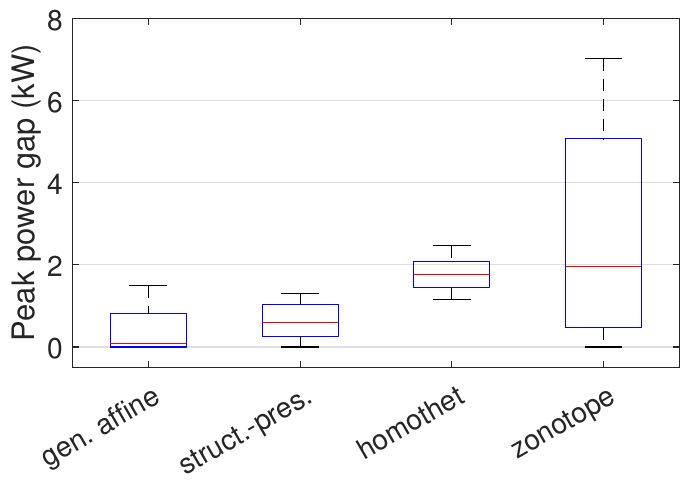}\label{fig:power_gap}}\\
    \vspace{15pt}
    \subfloat[Electricity cost suboptimality gap distributions]{\includegraphics[width=0.85\columnwidth]{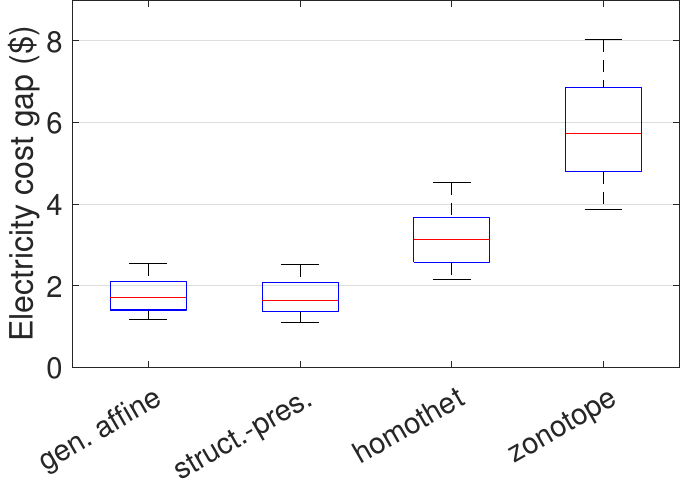}\label{fig:cost_gap}}
    \caption{Suboptimality gap distributions associated with each approximation method, as applied to the peak power minimization \eqref{eq:power_min} and electricity cost minimization \eqref{eq:cost_min} problems.
    The whiskers delimit the interdecile range, the box delimits the interquartile range, and the red line represents the median of each distribution.}
    \label{fig:subopt_gap}
\end{figure}

\begin{figure}[htb!]
    \centering
    \subfloat[Aggregate net-energy constraints]{
    \includegraphics[width=\columnwidth,trim= 0 4.9cm 0 0, clip]{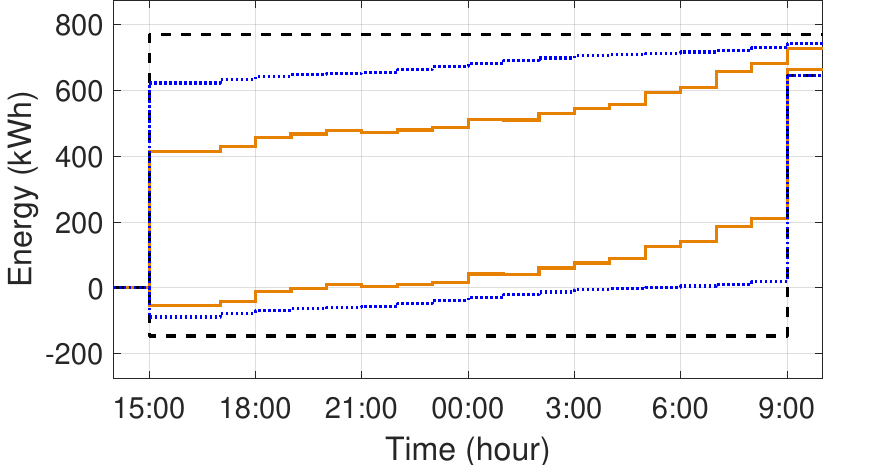}
    \label{fig:set_E}}\\
    \subfloat[Aggregate power constraints]{
    \includegraphics[width=\columnwidth,trim= 0 0.2cm 0 4.3cm, clip]{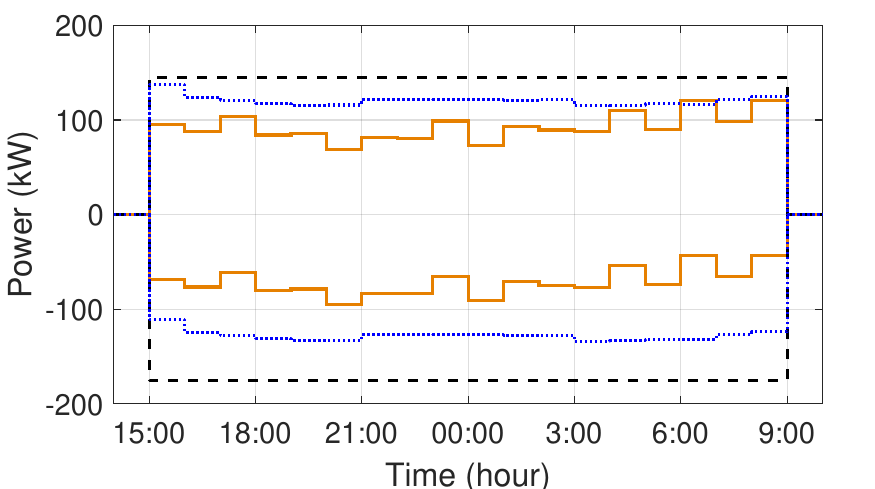}
    \label{fig:set_P}}\\
    \caption{Inner approximations of the aggregate flexibility set computed from data sampled on May 3, 2019. Depicted are the homothet-based inner approximation (solid orange lines), the structure-preserving inner approximation (dotted blue lines),  and the outer approximation $N\Uset_0$ (dashed black lines).}
    \label{fig:agg_set_approx}
\end{figure}

\begin{figure}[htb!]
    \centering
    \subfloat[Aggregate net-energy profiles]{
    \includegraphics[width=\columnwidth,trim= 0 4.9cm 0 0, clip]{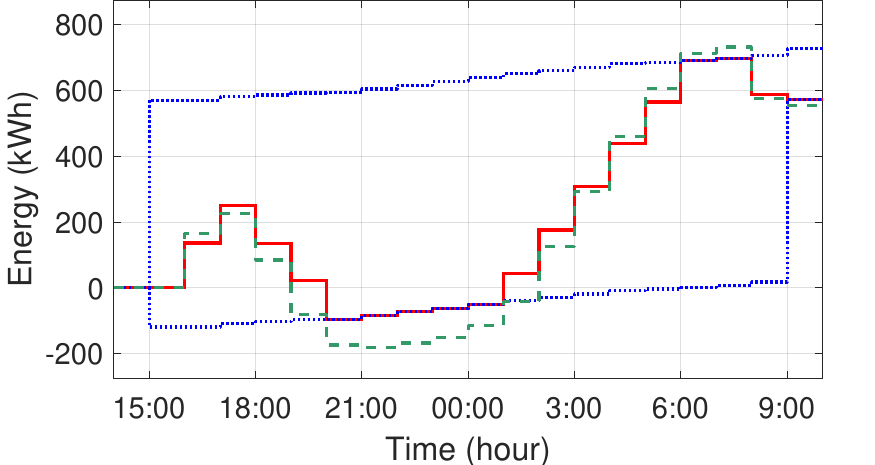}
    \label{fig:cost_E}}\\
    \subfloat[Aggregate power profiles]{
    \includegraphics[width=\columnwidth,trim= 0 0.2cm 0 4.3cm, clip]{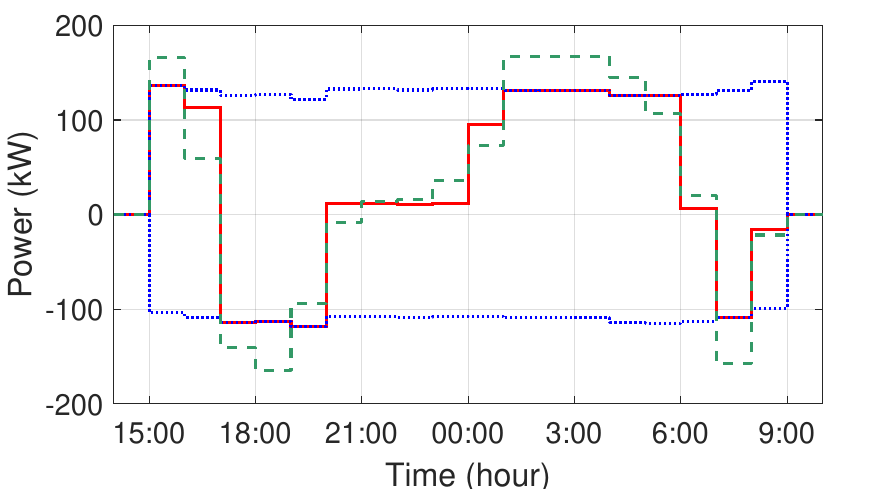}
    \label{fig:cost_P}}\\
     \hspace*{0.077in}\subfloat[Energy price profile]{\includegraphics[width=\columnwidth,trim= 0 0.2cm 0 4.3cm, clip]{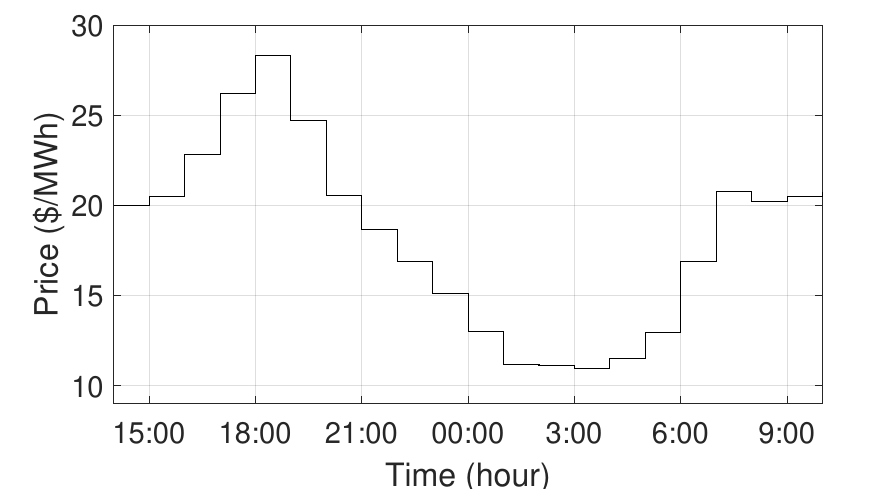}\label{fig:cost}}
    \caption{Control of EV aggregation to minimize electricity cost for data sampled on Oct. 31, 2019. (a), (b) Depicted are the optimal energy and power trajectories solving \eqref{eq:cost_min} (dashed green lines), the suboptimal energy and power trajectories  based on the structure-preserving approximation (solid red lines), and the energy and power limits associated with the structure-preserving approximation (dotted blue lines). (c) The NYISO Central Zone day-ahead energy prices for Oct. 31, 2019.}
    \label{fig:cost_soln}
\end{figure}

\begin{figure}
    \centering
    \includegraphics[width=.9\columnwidth]{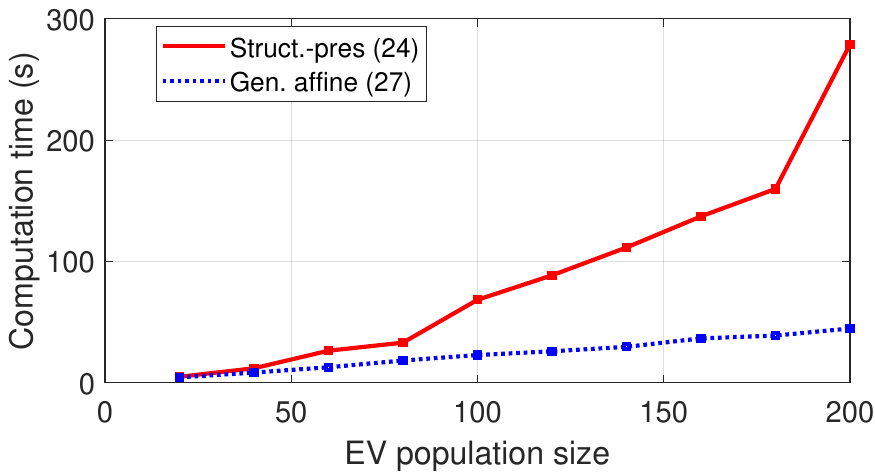}
    \caption{Inner approximation computation times (averaged over five independent trials) versus EV population size $N$ for the structure-preserving and general affine inner approximation methods.}
    \label{fig:compTimes}
\end{figure}

Figure \ref{fig:subopt_gap} shows that the general affine approximation method consistently outperforms the other three methods in the peak power minimization problem, achieving suboptimality gaps that are nearly zero for over half of the days in  the six-month span of the data set.
When applied to the cost minimization problem, we observe that the structure-preserving and general affine approximation methods outperform the homothet-based and zonotope-based approximation methods for a large majority of the days.
Compared to the other three methods studied, the zonotope-based approximation method exhibits considerably more day-to-day variability in performance, as shown by the larger variance in its suboptimality gap distributions.
This may stem from an incompatibility between the centrally-symmetric geometry of zonotopes and the asymmetric geometry of the individual flexibility sets.

In Figure \ref{fig:agg_set_approx}, we plot the power and net-energy limits associated with the homothet-based inner approximation, the structure-preserving inner approximation, and the outer approximation $N\Uset_0$, based on data from a randomly selected day.
For these data, the structure-preserving approximation provides significantly more flexibility than the homothet-based approximation. That being said, there is a nontrivial gap between the power and net-energy limits of structure-preserving approximation and those of the outer approximation. This may be indicative of conservatism in the structure-preserving method, suggesting that while it improves upon the homothet-based method, it may not fully capture the true aggregate flexibility set.

For another arbitrarily selected day, we depict in Figure \ref{fig:cost_soln} the boundaries associated with the structure-preserving inner approximation (dotted blue lines), along with the corresponding power and net-energy profiles (solid red lines) obtained when solving the electricity cost minimization problem \eqref{eq:cost_min} with this inner approximation.
As one might expect, we observe cycles of charging and discharging, which serve to capitalize on multiple inter-temporal price arbitrage opportunities over the course of the day.
The optimal power and net-energy profiles (dashed green lines) follow the same trend, but exceed the boundaries associated with the structure-preserving approximation.
This indicates that there is, in fact, some conservatism associated with the structure-preserving inner approximation.

We close by examining the behavior of computation times associated with the approximation methods proposed in this paper as a function of the EV population size $N$.
We initially sample 20 individual flexibility sets and  increase the population size incrementally, sampling 20 additional individual flexibility sets at each step, ranging from $N=20$ to $N=200$ sets.
The flexibility sets are sampled using the parameters described in Table \ref{tab:data_gen}.
For each value of $N$, we compute inner approximations to the aggregate flexibility set using 
both the structure-preserving and general affine inner approximation methods proposed in this paper.
For the structure-preserving inner approximation method, we record the time required to solve the LP in \eqref{eq:LP inner approx}.
For the general affine inner approximation method, we record the time required to solve the sequence of LPs in \eqref{eq:decomp 1}, which entails solving one LP per individual flexibility set $i \in \Ncal$. 
The LPs were solved using CVX (version 2.2) in MATLAB \cite{cvx}, using the MOSEK solver (version 9.1.9). 
A laptop with an AMD Ryzen 7 4700U processor and 16 GB of RAM was used for all computations.

In Fig. \ref{fig:compTimes}, we plot the resulting computation times (averaged over five independent trials) as a function of the population size $N$ for each method. The computation time required by the structure-preserving approximation method appears to scale super-linearly with $N$.
It may be possible to reduce these solve times by utilizing the Dantzig-Wolfe decomposition algorithm to take advantage of the block-angular sparsity structure in the linear program \eqref{eq:LP inner approx}.
In contrast, the computation time for the general affine approximation method scales linearly with the number of EVs, as one might expect.
It should also be noted that the general affine approximation method is trivially parallelizable, which could further decrease computation times by a factor of $N$. This ability to parallelize can significantly enhance the method's scalability to much larger EV populations.

%% file: Conclusion.tex
\section{Conclusion} \label{sec:conclusion}

    In this paper, we presented novel linear programming-based methods to compute inner approximations of the Minkowski sum of heterogeneous EV flexibility sets. By restricting the class of approximating sets to those which can be expressed as affine transformations of a given convex polytope (termed the base set), we showed how to conservatively approximate the resulting inner optimization problems as linear programs that scale polynomially with the number and dimension of the individual flexibility sets. The proposed approximation methods were shown to generalize and improve upon the approximation accuracy of related methods in the literature. 
    We also provided an efficient disaggregation method to decompose any aggregate charging profile within the proposed inner approximations into a collection of individually feasible charging profiles, without requiring the solution of another optimization problem to perform the disaggregation.
    
    As a direction for future research, we  intend to generalize the methods developed in this paper to account for lossy EV charging dynamics with energy leakage and energy conversion inefficiencies.
    It would also be interesting to extend these approximation techniques to incorporate distribution network capacity constraints  that impact how EVs can be aggregated across a large network. In such settings, a hierarchical approach to aggregation may prove effective in handling localized constraints at different levels of the distribution network.

%% file: Appendix.tex
\appendices

\section{Nomenclature} \label{app:nomenclature}
\renewcommand{\nomname}{}

\printnomenclature